\documentclass[11pt]{article}

\usepackage[margin = 1in]{geometry}
\usepackage{graphicx}              
\usepackage{amsmath}               
\usepackage{amsfonts}              
\usepackage{amsthm}                
\usepackage{amssymb}
\usepackage{mathrsfs}
\usepackage{url}
\usepackage{xcolor}
\usepackage{hyperref}
\usepackage[inline]{enumitem}
\usepackage{authblk}
\usepackage{tcolorbox}
\usepackage{mathtools}
\usepackage{bm}
\usepackage{dsfont}
\usepackage{tikz}
\usetikzlibrary{cd}

\hypersetup{
	unicode = true,
	colorlinks = true,
	citecolor = blue,
	filecolor = blue,
	linkcolor = blue,
	urlcolor = blue,
	pdfstartview = {FitH},
}

\theoremstyle{plain}
\newtheorem{theorem}{Theorem}[section]
\newtheorem{lemma}[theorem]{Lemma}
\newtheorem{corollary}[theorem]{Corollary}
\newtheorem{proposition}[theorem]{Proposition}

\theoremstyle{definition}
\newtheorem{definition}[theorem]{Definition}

\newtheorem{algo-thm}{Algorithm}
\newtheorem{remark}{Remark}

\newcommand{\Romnum}[1]{\uppercase\expandafter{\romannumeral #1}}

\DeclareMathOperator{\groupofend}{End} 
\DeclareMathOperator{\tr}{Tr} 
\DeclareMathOperator{\norm}{N} 

\DeclareMathOperator{\GL}{GL}

\DeclareMathOperator{\poly}{poly}

\DeclareMathOperator{\E}{\mathbb{E}}

\DeclarePairedDelimiter{\abs}{\lvert}{\rvert}
\DeclarePairedDelimiter{\ket}{\lvert}{\rangle}
\DeclarePairedDelimiter{\bra}{\langle}{\rvert}
\DeclarePairedDelimiterX{\braket}[2]{\langle}{\rangle}{#1 \delimsize\vert #2}
\DeclarePairedDelimiter{\lrang}{\langle}{\rangle}
\DeclarePairedDelimiter{\opnorm}{\lVert}{\rVert}

\def\Q{\mathbb{Q}}
\def\C{\mathbb{C}}

\def\Z{\mathbb{Z}}
\def\F{\mathbb{F}}

\def\X{\mathcal{X}}

\title{How to Sample From The Limiting Distribution of a \\ Continuous-Time Quantum Walk}

\author{
	Javad Doliskani\thanks{Department of Computer Science, Ryerson University,
	{\tt javad.doliskani@ryerson.ca}.}
}

\date{}
\allowdisplaybreaks

\begin{document}
\maketitle

\begin{abstract}
    We introduce $\varepsilon$-projectors, using which we can sample from limiting distributions of continuous-time quantum walks. The standard algorithm for sampling from a distribution that is close to the limiting distribution of a given quantum walk is to run the quantum walk for a time chosen uniformly at random from a large interval, and measure the resulting quantum state. This approach usually results in an exponential running time. 

    We show that, using $\varepsilon$-projectors, we can sample exactly from the limiting distribution. In the black-box setting, where we only have query access to the adjacency matrix of the graph, our sampling algorithm runs in time proportional to $\Delta^{-1}$, where $\Delta$ is the minimum spacing between the distinct eigenvalues of the graph. In the non-black-box setting, we give examples of graphs for which our algorithm runs exponentially faster than the standard sampling algorithm. 
\end{abstract}

\vfill

\tableofcontents

\newpage

\section{Introduction}
\label{sec:intro}

Continuous-time quantum walks, first considered by Farhi and Gutmann \cite{farhi1998quantum}, are quantum analogues of continuous-time classical random walks. The dynamics of a continuous-time classical walk on an undirected graph $\Gamma$ is described by the differential equation
\begin{equation}
    \label{equ:clss-walk}
    \frac{d}{dt}p(t) = Lp(t)
\end{equation}
where $p(t)$ is the state of the walk at time $t$ and $L$ is the Laplacian of $\Gamma$. The entries of $p(t)$ are indexed by the set of vertices of $\Gamma$. In the quantum walk on $\Gamma$, equation \eqref{equ:clss-walk} is replaced by the Schr\"{o}dinger equation
\begin{equation}
    \label{equ:Schrodinger}
    i\frac{d}{dt}\ket{\psi_t} = H\ket{\psi_t}
\end{equation}
where the Hamiltonian $H = L$, and $\ket{\psi_t}$ is a quantum state whose amplitudes encode a probability distribution. Quantum walks have found many applications in quantum computing and quantum information. It was shown by Childs \cite{childs2009universal} that universal quantum computation can be implemented using quantum walks on low degree graphs. There are many algorithms based on quantum walk that achieve polynomial speedup over classical algorithms, e.g., \cite{childs2004spatial, farhi2007quantum, ambainis2010any, apers2021quadratic}. There are also black-box problems for which quantum walks achieve exponential speedup over classical algorithms \cite{childs2003exponential, childs2007quantum}.

A classical continuous-time random walk has a unique stationary distribution, assuming its underlying Markov chain is irreducible \cite{levin2017markov}. Regardless of the initial state, the walk converges to this stationary distribution as $t \rightarrow \infty$, and therefore, the stationary distribution is the same as the limiting distribution. This, however, does not hold for quantum walks, since quantum evolutions are unitary and preserve distance. Nevertheless, one can define a time-averaged probability distribution of a quantum walk by choosing a time $t \in [0, T]$ uniformly at random, running the walk for a total time $t$ and measuring the resulting state. When $T \rightarrow \infty$, this distribution converges to a limiting distribution which is what we will consider in this paper. This limiting distribution generally depends on the initial state of the walk.

If the graph $\Gamma$ is connected and simple, i.e., has no self loops and multiple edges, then the limiting distribution of the classical walk on $\Gamma$ is always the uniform distribution. In contrast, the limiting distribution of the quantum walk on $\Gamma$ depends on the initial state and is often not uniform. For example, the quantum walk on the hypercube \cite{moore2002quantum} or the Symmetric group \cite{gerhardt2003continuous}, or more generally $G$-circulant graphs \cite{adamczak2003note} does not converge to the uniform distribution. 

The mixing time $M_\delta$ of a quantum (or classical) walk is the minimum time after which the distribution of the walk is within distance $\delta$ of the limiting distribution. The mixing time of a classical walk depends inversely on the spectral gap of the transition matrix of the walk, while for a quantum walk, the mixing time depends inversely on the minimum gap between all pairs of distinct eigenvalues of the Hamiltonian $H$. The mixing time of quantum walks have been studied for specific graphs such as hypercubes, cycles and lattices \cite{moore2002quantum, fedichkin2005mixing, richter2007almost}, and for Erd\H{o}s-R\'{e}nyi random graphs \cite{chakraborty2020fast, chakraborty2020analog}. The bounds on the quantum mixing time for some graphs imply quadratic speedup over classical walks, while for some other graphs these bounds are larger than their classical counterparts.

The problem of sampling from the limiting distribution of a classical walk is an important problem and has been the focus of much research. The underlying randomness strategy in many algorithms of practical interest reduces to sampling from a limiting distribution. Much like in the classical case, sampling from the limiting distribution of a quantum walk is of both practical and theoretical interest. For example, Richter \cite{richter2007almost} proposed a ``double-loop'' quantum walk algorithm for sampling from a distribution that is close to the uniform distribution over a given graph $\Gamma$. The inner loop in Richter's algorithm samples from a distribution that is close to the limiting distribution $\Pi$ of the quantum walk. Therefore, an efficient algorithm for sampling from $\Pi$ results in an efficient algorithm for sampling uniformly from the vertex set of $\Gamma$. The black-box graph problem proposed by Childs et al.\;\cite{childs2003exponential}, which is interesting from a theoretical perspective, is based on sampling from the limiting distribution of a quantum walk. They prove that no classical algorithm can efficiently solve the proposed problem. 

Given the importance of sampling from the limiting distribution of a quantum walk, it is natural to ask whether there are efficient algorithms for sampling from such distributions, at least for specific graphs. There has not been much research explicitly addressing this question. The standard way of sampling from the limiting distribution of a quantum walk is by mixing: set a large value for $T$, run the walk for a uniformly random time $t \in [0, T]$, and measure. Chakraborty et al.\;\cite{chakraborty2020analog} considered sampling over the Erd\H{o}s-R\'{e}nyi graphs by mixing. They obtained an upper bound on the mixing time of these graphs through analyzing their spectrum. Their bound implies an exponential time sampling algorithm.

\subsection{Sampling without mixing}

In this work, we propose an algorithm for sampling from the limiting distribution of a given continuous-time quantum walk, that is not based on mixing. The idea behind our algorithm is to uniquely ``tag'' the eigenspaces of the adjacency matrix using polynomially long binary strings. More precisely, give a graph $\Gamma$ with $N$ vertices and adjacency matrix $A$, let $\ket{\phi_j}$ be an eigenstate of $A$ that belongs to an eigenspace $\X_j$ of $A$. Then, we will see that sampling from the limiting distribution $\Pi$ on $\Gamma$ reduces to performing the transform
\begin{equation}
    \label{equ:tag-op-intro}
    \ket{0}\ket{\phi_j} \mapsto \ket{t_j}\ket{\phi_j}
\end{equation}
where $t_j$ is a string of length $\poly(\log N)$ that uniquely identifies $\X_j$.  

To perform the transform \eqref{equ:tag-op-intro}, we introduce the general idea of $\varepsilon$-projectors. Informally, an $\varepsilon$-projector for the adjacency matrix $A$ is a set of hermitian matrices that have the same eigenspaces as $A$, and can be efficiently simulated as Hamiltonians. Moreover, the set of matrices in an $\varepsilon$-projector have to satisfy a separation condition with respect to their eigenvalues. A set of matrices satisfying such a separation condition is called an $\varepsilon$-separated set. A specific $\varepsilon$-projector was first implicitly used by Kane, Sharif and Silverberg \cite{kane2021quantum} for constructing a quantum money scheme based on quaternion algebras. Their $\varepsilon$-projector is a set of sparse Brandt matrices, which they use to verify an alleged bill.

\paragraph{Technique.}
Given an $\varepsilon$-projector $\mathcal{A}$ for $A$, we use phase estimation to store, in a separate register, an estimate of the eigenvalues of each operator in $\mathcal{A}$. More precisely, we perform the transform
\[ \ket{0}\ket{\phi_j} \mapsto \ket{\tilde{\lambda}_{1, j}} \ket{\tilde{\lambda}_{2, j}} \cdots \ket{\tilde{\lambda}_{r, j}} \ket{\phi_j} \]
where the $\tilde{\lambda}_{k, j}$ are approximate eigenvalues corresponding to the eigenstate $\ket{\phi_j}$. It follows from the $\varepsilon$-separatedness of $\mathcal{A}$ that the vectors $\bm{\lambda}_j = (\tilde{\lambda}_{1, j}, \dots, \tilde{\lambda}_{r, j})$ uniquely identify the eigenspaces of $A$. This means the binary representation of the $\bm{\lambda}_j$ can be used as the binary strings $t_j$ in \eqref{equ:tag-op-intro}. Therefore, efficient sampling from the limiting distribution $\Pi$ reduces to finding a good $\varepsilon$-projector for $A$. In general, a good $\varepsilon$-projector for $A$ is an $\varepsilon$-projector for which $\varepsilon^{-1}$ and $r$ are both at most $\poly(\log N)$. In specific cases where the operators in the $\varepsilon$-projector can be simulated efficiently even for large powers, it is not necessary for $\varepsilon^{-1}$ to be bounded by $\poly(\log N)$.

\paragraph{Black-box vs non-black-box.}
When the graph $\Gamma$ is given as a black-box, we can only see the structure of $\Gamma$ locally. In other words, we can only access the nonzero entries of each row of the adjacency matrix $A$ of $\Gamma$ through queries to an oracle. When we are restricted to query access to $A$, the global structure of $\Gamma$ is not known, and we do not have any other information on $A$ as an operator. In this case, we are essentially left with one choice for an $\varepsilon$-projector for $A$: $A$ itself. Consequently, we always have $\varepsilon \le \Delta$, where $\Delta$ is the minimum distance between any two distinct eigenvalues of $A$. The complexity of our sampling algorithm will then be bounded below by a multiple of $\Delta^{-1}$. 

In the non-black-box setting, we often have some knowledge of the global structure of $\Gamma$ that enables us to find nontrivial $\varepsilon$-projectors for $A$. In this paper, we give two examples of graphs for which we can find good $\varepsilon$-projectors, Winnie Li graphs and Supersingular Isogeny graphs.

Winnie Li graphs \cite{li1992character} are special cases of quasi-abelian graphs which are a subclass of Cayley graphs. We give a general strategy for finding $\varepsilon$-projectors for quasi-abelian graphs, and apply it to Winnie Li graphs. Without the use of an $\varepsilon$-projector, one could sample from the limiting distribution of quantum walks on Winnie Li graphs using two different methods. The first method is by mixing, which takes exponential time because the eigenvalues of the (normalized) adjacency matrix are very close. The second method is to use the quantum Fourier transform. For that, we need to be able to approximate the eigenvalues of $A$ efficiently. When the dimension of the underlying space is odd, these eigenvalues are multiples of some exponential sums called Kloosterman sums. There is no known efficient classical or quantum algorithm for approximating these sums. We will see that using a specific $\varepsilon$-projector, we can efficiently sample from the limiting distributions on these graphs.

A supersingular isogeny graphs is a regular graph in which the set of vertices is the set of supersingular elliptic curves and the edges are isogenies between these curves. Isogeny graphs have found many applications in cryptography \cite{de2017mathematics, de2018exploring}. The adjacency matrices of these graphs are called Hecke operators. The minimum distance between the eigenvalues of a Hecke operator is exponentially small, so the quantum mixing time for these graphs is exponentially large. It is known that the set of Hecke operators form a commutative algebra over $\C$. Using this fact, and assuming some standard heuristics, we will see that a small set of these operators form an $\varepsilon$-projector with high probability. Using this $\varepsilon$-projector, we can efficiently sample from the limiting distribution on these graphs. As an application, the sampling algorithm can be used to generate \textit{honest hard curves}. There is no known classical algorithm for efficiently generating such curves.


\section{Preliminaries}
\label{sec:prelim}

\subsection{Continuous-time quantum walk}
\label{sec:prlm-qwalk}

Let $\Gamma = (V, E)$ be an undirected graph with $N = \abs{V}$ vertices, and let $\X = \C^V$ be the complex euclidean space with basis $V$. We will refer to this basis as the vertex basis and denote its elements by $\ket{v}$, $v \in V$. Let $A$ be the adjacency matrix of $\Gamma$. The continuous-time quantum walk on $\Gamma$ is described by the differential equation \eqref{equ:Schrodinger} where the Hamiltonian is the Laplacian of $\Gamma$. Another common choice (which we also use in this paper) for the Hamiltonian of the walk is the adjacency matrix $A$. Then, the continuous-time quantum walk on $G$ at time $t$ is defined by the operator
\[ W(t) = e^{iAt} \]
on $\X$. For an initial quantum state $\ket{\psi_0}$ and a real number $T > 0$, define the following probability distribution on $V$: choose $t \in [0, T]$ uniformly at random, evolve the state $\ket{\psi_0}$ under $W(t)$, i.e., compute $W(t)\ket{\psi_0}$, and measure the resulting state in the vertex basis. The probability of measuring a vertex $v \in V$ is
\begin{equation}
    \label{equ:avg-dist}
    P_T(v \vert \psi_0) = \frac{1}{T} \int_0^T \abs{\braket{v}{W(t) \vert \psi_0}}^2 dt.
\end{equation}

Let $\{\ket{\phi_j}\}_{1 \le j \le N}$ be a set of eigenstates of $A$ that form an orthonormal basis for $\X$, and let $\{\lambda_j\}_{1 \le j \le N}$ be the set of corresponding eigenvalues. Let $\{ \X_j \}_{1 \le j \le M}$, where $M \le N$, be the set of eigenspaces of $A$, and define $I_j = \{k : \ket{\phi_k} \in \X_j \}$. Therefore, $I_j$ is the set of indices $k$ for which the eigenstates $\ket{\phi_k}$ correspond to the same eigenvalue. A straightforward calculation shows that
\begin{align}
    P_T(v \vert \psi_0) =
    & \sum_{j = 1}^M \abs[\Big]{\sum_{k \in I_j} \braket{v}{\phi_k} \braket{\phi_k}{\psi_0}}^2 + \nonumber \\
    & \sum_{k, \ell = 1 : \lambda_k  \ne \lambda_\ell}^N \braket{v}{\phi_k} \braket{\phi_k}{\psi_0} \braket{\psi_0}{\phi_\ell} \braket{\phi_\ell}{v} \frac{e^{i(\lambda_k - \lambda_\ell)T} - 1}{i(\lambda_k - \lambda_\ell)T} \label{equ:walk-dist}
\end{align}
Letting $T \rightarrow \infty$, the second term in the above expansion will vanish, and we get the distribution
\begin{equation}
    \label{equ:lim-dist}
    P_\infty(v \vert \psi_0) := \lim_{T \rightarrow \infty} P_T(v \vert \psi_0) = \sum_{j = 1}^M \abs[\Big]{\sum_{k \in I_j} \braket{v}{\phi_k} \braket{\phi_k}{\psi_0}}^2.
\end{equation}
This is called the \textit{limiting distribution} of the quantum walk $W(t)$. Given a real number $\delta \ge 0$, the mixing time $M_\delta$ of the walk $W(t)$, with respect to the initial state $\ket{\psi_0}$, is defined as
\begin{equation}
    \label{equ:mixing-time}
    M_\delta = \min \{T' : \opnorm{P_T(\cdot \vert \psi_0) - P_\infty(\cdot \vert \psi_0)}_1 \le \delta, \forall T \ge T' \},
\end{equation}
where $P_T(\cdot \vert \psi_0)$ and $P_\infty(\cdot \vert \psi_0)$ are the probability vectors defined by \eqref{equ:avg-dist} and \eqref{equ:lim-dist}, respectively. Denote by $\Delta$ the minimum distance between all pairs of distinct eigenvalues of $A$, i.e.,
\begin{equation}
    \label{equ:eigen-dist}
    \Delta = \min_{\lambda_k \ne \lambda_\ell} \abs{\lambda_k - \lambda_\ell}, 1 \le k, \ell \le N.
\end{equation}
Using the same analysis as in \cite{aharonov2001quantum}, it can be shown that
\begin{equation}
    \label{equ:mixing-distnc}
    \opnorm{P_T(\cdot \vert \psi_0) - P_\infty(\cdot \vert \psi_0)}_1 \le \frac{2 \ln M + 2}{T\Delta}.
\end{equation}
A proof of this bound is given in Appendix \ref{sec:mixing:proof} for completeness. From the definition of $M_\delta$ and the bound \eqref{equ:mixing-distnc}, we see that 
\begin{equation}
    \label{equ:mixing-bound}
    M_\delta \le \frac{2 \ln M + 2}{\delta \Delta}.
\end{equation}

\subsection{Representation theory}

For an introduction to representation theory see \cite{curtis1966representation, serre1977linear}. Let $V$ be a $\C$-vector space of finite dimension, and let $\GL(V)$ be the group of automorphisms of $V$. Let $G$ be a finite group. A linear representation of $G$ in $V$ is a homomorphism of groups $\rho: G \rightarrow \GL(V)$. The degree of $\rho$, denoted by $d_\rho$, is the dimention of $V$ as a $\C$-vector space. The character of $\rho$ is the function $\chi_\rho: G \rightarrow \C$ defined by $\chi_\rho(a) = \tr \rho(a)$.  A morphism of representations $\rho_1: G \rightarrow \GL(V_1)$ and $\rho_2: G \rightarrow \GL(V_2)$ is a $\C$-linear map $\phi: V_1 \rightarrow V_2$ such that for every $a \in G$ the diagram
\begin{center}
    \begin{tikzcd}[row sep = large, column sep = large]
        V_1 \arrow[r, "\rho_1(a)"] \arrow[d, "\phi"] & V_2 \arrow[d, "\phi"] \\
        V_2 \arrow[r, "\rho_2(a)"] & V_2
    \end{tikzcd}
\end{center}
commutes. The representations $\rho_1$ and $\rho_2$ are said to be isomorphic if $\phi$ is an isomorphism. A subrepresentation of $\rho$ is a representation $\rho^W: G \rightarrow \GL(W)$ where $W \subseteq V$ is linear subspace. A representation that has no subrepresentations except for $W = 0, V$ is called an irreducible representation. We denote by $\widehat{G}$ the set of isomorphism classes of irreducible representations of $G$. Any representation $\rho$ of $G$ can be decomposed as a direct sum of irreducible representations: If $\varrho_1, \dots, \varrho_k$ is a complete set of irreducible representations of $G$ then $\rho = n_1 \varrho_1 \oplus \cdots \oplus n_k \varrho_k$ for some integers $n_j \ge 0$. Here, $n_j\varrho_j$ means a direct sum of $n_j$ copies of $\varrho_j$.

Let $V = \C^G$, i.e., the space of functions $f: G \rightarrow \C$. The (left) regular representation of $G$ is a representation $\rho_\mathrm{reg}$ in $V$ defined by $\rho_\mathrm{reg}(s) f(a) = f(s^{-1}a)$ for any function $f \in V$. The regular representation decomposes as
\begin{equation}
    \label{equ:reg-rep}
    \rho_\mathrm{reg} \cong \bigoplus_{\varrho \in \widehat{G}} d_\varrho\varrho,
\end{equation}
which also shows that $\sum_{\varrho} d_\varrho^2 = \abs{G}$.

A representation $\rho$ of $G$ is called unitary if $\rho(a)$ is a unitary matrix for all $a \in G$. Given any representation $\rho$ of $G$, there always exists an inner product on $V$ with respect to which $\rho$ is unitary. Therefore, in this paper, we assume that all representations are unitary. In particular, any unitary representation can be decomposed as a sum of unitary irreducible representations.

The Fourier transform of a function $f: G \rightarrow \C$ at a representation $\varrho \in \widehat{G}$ is defined by
\[ \widehat{f}(\varrho) = \sqrt{\frac{d_\varrho}{\abs{G}}} \sum_{a \in G} \varrho(a) f(a). \]
The Fourier transform of $f$ is given by $\oplus_\varrho \widehat{f}(\varrho)$. The quantum Fourier transform of a state $\ket{\psi} = \sum_{x \in G} \alpha_x \ket{x}$ is given by
\begin{equation}
    \label{equ:qft-g}
    \ket{\widehat{\psi}} = \sum_{\varrho \in \widehat{G}} \; \sum_{1 \le j, k \le d_\varrho} \widehat{\alpha}(\varrho)_{j, k} \ket{\varrho, j, k}
\end{equation}
where $\alpha: G \rightarrow \C$ is defined by $\alpha(x) = \alpha_x$, and $\widehat{\alpha}(\varrho)_{j, k}$ is the $(i, j)$ entry of the matrix $\widehat{\alpha}(\varrho)$.


\section{Sampling Using \texorpdfstring{$\bm{\varepsilon}$}{epsilon}-Projectors}
\label{sec:sampling}

Let $\Gamma = (V, E)$ be a graph with $N$ vertices and let $A$ be the adjacency matrix of $\Gamma$. Assume the same notation as in Section \ref{sec:prlm-qwalk}. A closer look at the sum in \eqref{equ:lim-dist} suggests the following simple approach to sampling from the limiting distribution $P_\infty$. Since $\{\ket{\phi_j}\}$ is an orthonormal basis for $\X$, given any initial state $\ket{\psi_0}$, we can always write
\[ \ket{\psi_0} = \sum_{j = 1}^N \braket{\phi_j}{\psi_0} \ket{\phi_j}. \]
Suppose we have a quantum algorithm $Q$ that can uniquely ``tag'' the eigenspaces of $A$ in the above superposition, using an extra register. More precisely, $Q$ performs the following operation
\begin{equation}
    \label{equ:tag-op}
    \sum_{j = 1}^N \braket{\phi_j}{\psi_0} \ket{0} \ket{\phi_j} \longmapsto \sum_{j = 1}^N \braket{\phi_j}{\psi_0} \ket{t_j} \ket{\phi_j} = \sum_{j = 1}^M \sum_{k \in I_j} \braket{\phi_k}{\psi_0} \ket{t_j} \ket{\phi_k}
\end{equation}
where the strings $t_j$ are unique with respect to the eigenspaces of $A$. If we measure the second register in the vertex basis, we obtain a vertex $v \in V$ with the probability given by \eqref{equ:lim-dist}.

A naive choice for the tags $t_j$ are the eigenvalues of $A$. These eigenvalues can be approximated using phase estimation on the walk operator $W(t)$. However, to be able to uniquely identify the eigenspaces of $A$, one might need to compute the eigenvalues with exponential accuracy. Any such computation generally takes exponential time unless $W(t)$ can be applied efficiently for exponentially large $t$. In particular, if we treat $A$ as a black-box, the complexity of performing \eqref{equ:tag-op} is going to be exponential in $t$. Therefore, any successful attempt at efficiently performing \eqref{equ:tag-op} will require some extra information or assumptions on $A$.

In the following we present the main idea of the paper, an algorithm for performing \eqref{equ:tag-op} that uses a specific set of operators that commute with $A$. We call such a set of operators an $\varepsilon$-\textit{projector}. For many classes of graphs, we can find $\varepsilon$-projectors that enable us to efficiently perform \eqref{equ:tag-op}. We need to adapt the definition of an $\varepsilon$-separated set from \cite{kane2021quantum} to a set of operators.

\begin{definition}
    For an integer $r > 0$, let $\mathcal{A} = \{A_j\}_{1 \le j \le r}$ be a set of hermitian operators, acting on $\X$, that have the same eigenspaces. For an eigenstate $\ket{\phi_j}$, let $\lambda_{1, j}, \lambda_{2, j}, \dots, \lambda_{r, j}$ be the eigenvalues of the operators $A_1, A_2, \dots, A_r$ associated with $\ket{\phi_j}$, respectively. Define the vector $\bm{\lambda}_j = (\lambda_{1, j}, \lambda_{2, j}, \dots, \lambda_{r, j})$ for each $j = 1, \dots, N$. For a real number $\varepsilon > 0$, the set of operators $\mathcal{A}$ is said to be $\varepsilon$-separated if
    \[ \opnorm{\bm{\lambda}_j - \bm{\lambda}_k}_2 \ge \varepsilon, \text{ for all } \bm{\lambda}_j \ne \bm{\lambda}_k, 1 \le j, k \le N. \]    
\end{definition}

\begin{definition}
    \label{def:eps-proj}
    Let $A$ be a hermitian operator on $\X$. An $\varepsilon$-projector for $A$ is an $\varepsilon$-separated set $\mathcal{A} = \{A_j\}_{1 \le j \le r}$ such that for all $j = 1, \dots, r$
    \begin{itemize}[itemsep = 0mm, parsep = 1mm, topsep = 1mm]
    \item the walk $e^{iA_jt}$ can be performed in $F(t) \poly(\log N)$ operations, where $F(t) \in O(t)$, and
    \item $A_j$ has the same eignenspaces as $A$.
    \end{itemize}
\end{definition}

The function $F(t)$ in Definition \ref{def:eps-proj} determines how efficient the walks $e^{iA_jt}$ can be performed for different values of $t$. For an $\varepsilon$-projector we require that $F(t)$ be bounded above by a linear function in $t$. Also, here operations refer to elementary quantum gate operations. We now give an algorithm for sampling from the limiting distribution of the quantum walk on $\Gamma$. The algorithm takes as input an $\varepsilon$-projector for the adjacency matrix $A$.

\begin{algo-thm}[Sampling]
    \label{alg:sampling} \
    \begin{description}[font = \normalfont\itshape, topsep = 0pt, itemsep = 0pt, parsep = 0pt]
    \item[Input:] An adjacency matrix $A$ of a graph $\Gamma = (V, E)$, an $\varepsilon$-projector $\mathcal{A} = \{A_j\}_{1 \le j \le r}$ for $A$, an initial state $\ket{\psi_0} \in \X$.
    \item[Output:] A sample from the limiting distribution of the walk $W(t) = e^{iAt}$ on $\Gamma$.
    \end{description}
    \begin{enumerate}[leftmargin = *, topsep = 0pt, itemsep = 0pt, parsep = 0pt]
    \item  Perform phase estimation on the unitaries $e^{iA_1}, \dots, e^{iA_r}$ and the input state $\ket{\psi_0}$ with accuracy $\varepsilon / 2\sqrt{r}$, and store the approximate phases in extra registers. Denote by $\tilde{\lambda}_{k, j}$ the approximation of the eigenvalue $\lambda_{k, j}$ of $A_k$ corresponding to the eigenstate $\ket{\phi_j}$. Then the resulting state of this step is
        \begin{equation}
            \label{equ:many-phase}
            \sum_{j = 1}^N \braket{\phi_j}{\psi_0} \ket{\tilde{\lambda}_{1, j}} \ket{\tilde{\lambda}_{2, j}} \cdots \ket{\tilde{\lambda}_{r, j}} \ket{\phi_j}.
        \end{equation}
        where $\abs{\lambda_{i, j} - \tilde{\lambda}_{i, j}} < \varepsilon / 2\sqrt{r}$ for all $i = 1, \dots, r$. If we group the content of the first $r$ registers as a vector $\tilde{\bm{\lambda}}_j$ then the state \eqref{equ:many-phase} can be written as
        \begin{equation}
            \label{equ:many-phase-g}
            \sum_{j = 1}^N \braket{\phi_j}{\psi_0} \ket{\tilde{\bm{\lambda}}_j} \ket{\phi_j}
        \end{equation}
    \item Measure the last register in the vertex basis.
    \item Return the measured vertex. 
    \end{enumerate}
\end{algo-thm}

Algorithm \ref{alg:sampling} proposes to use the (binary representations) of the phase vectors $\tilde{\bm{\lambda}}_j$ as the tags $t_j$ in \eqref{equ:tag-op}. The correctness of the algorithm follows from the next lemma.

\begin{lemma}
    The vectors $\tilde{\bm{\lambda}}_j$ uniquely determine the eigenspaces of $A$. More precisely, $\tilde{\bm{\lambda}}_j = \tilde{\bm{\lambda}}_k$ if and only if $j$ and $k$ both belong to $I_\ell$ for some $\ell$.
\end{lemma}
\begin{proof}
    Let $\bm{\lambda}_j = (\lambda_{1, j}, \lambda_{2, j}, \dots, \lambda_{r, j})$ be the vectors of the exact eigenvalues of $A_1, \dots, A_r$ corresponding to the eigenstate $\ket{\phi_j}$. Then for all $j = 1, \dots, N$ we have
    \begin{equation}
        \label{equ:tag-ineq}
        \opnorm{\bm{\lambda}_j - \tilde{\bm{\lambda}}_j}_2^2 = \sum_{i = 1}^r \abs{\lambda_{i, j} - \tilde{\lambda}_{i, j}}^2 < \frac{\varepsilon^2}{4} 
    \end{equation}
    where the last inequality follows from the bound $\abs{\lambda_{i, j} - \tilde{\lambda}_{i, j}} < \varepsilon / 2\sqrt{r}$. Now, suppose $\tilde{\bm{\lambda}}_j = \tilde{\bm{\lambda}}_k$ where $k \in I_h$ and $j \in I_\ell$ and $h \ne \ell$. Then
    \begin{align*}
        \opnorm{\bm{\lambda}_j - \bm{\lambda}_k}_2
        & = \opnorm{(\bm{\lambda}_j - \tilde{\bm{\lambda}}_j) - (\bm{\lambda}_k - \tilde{\bm{\lambda}}_k)}_2 \\
        & \le \opnorm{\bm{\lambda}_j - \tilde{\bm{\lambda}}_j}_2 + \opnorm{\bm{\lambda}_k - \tilde{\bm{\lambda}}_k}_2 \\
        & < \varepsilon  \tag*{by \eqref{equ:tag-ineq}}
    \end{align*}
    which contradicts the $\varepsilon$-separatedness of $\mathcal{A}$.
\end{proof}

The following theorem records the main result of this section.

\begin{theorem}
    \label{thm:main}
    Let $\Gamma = (V, E)$ be a graph with $N$ vertices and adjacency matrix $A$. Given an $\varepsilon$-projector $\mathcal{A} = \{A_j\}_{1 \le j \le r}$ for $A$ such that the walks $e^{iA_jt}$ can be performed in $F(t) \poly(\log N)$ operations, there is a quantum algorithm that can sample from the limiting distribution of the walk $W(t) = e^{iAt}$ in $O(r F(2\sqrt{r}\varepsilon^{-1}) \poly(\log N))$ operations.
\end{theorem}
\begin{proof}
    The time consuming part of the algorithm is the phase estimation for the operators $W_j = e^{iA_j}$ for $j = 1, \dots, r$. Each of these phase estimations is done with accuracy $\varepsilon / 2\sqrt{r}$, and therefore requires $O(F(2\sqrt{r}\varepsilon^{-1}) \poly(\log N))$ operations \cite[Chapter 7]{kaye2006introduction}. Since there are $r$ phase estimations, the claimed running time follows.
\end{proof}

From Theorem \ref{thm:main} we see that the complexity of Algorithm \ref{alg:sampling} is mostly determined by the ``quality'' of the given $\varepsilon$-projector $\mathcal{A}$. If the number $r$ of the operators in $\mathcal{A}$ and the separation parameter $\varepsilon$ are $\poly(\log N)$ and $1 / \poly(\log N)$, respectively, then the algorithm is efficient, i.e., runs in $\poly(\log N)$ operations. Otherwise, there is a natural trade-off between the sizes of the two parameters. Also note that, by definition, we always have $F(t) \in O(t)$ for any $\varepsilon$-projector, so the running time in Theorem \ref{thm:main} is always upper bounded by $O(r^{3 / 2} \varepsilon^{-1} \poly(\log N))$. 

An immediate special case of Theorem \ref{thm:main} is when the $\varepsilon$-projector is a singleton set $\mathcal{A} = \{B\}$, that is, when $r = 1$. If we only have black-box access to $B$ we can only perform the walk $e^{iBt}$ with a running time that scales linearly in $t$. In this case, the running time of Algorithm \ref{alg:sampling} scales linearly in $\varepsilon^{-1}$. On the other hand, if we can perform the walk $e^{iBt}$ with a running time that scales polynomially in $\log t$, then the running time of Algorithm \ref{alg:sampling} scales polynomially in $\log(\varepsilon^{-1})$. A lower bound for the running time of Algorithm \ref{alg:sampling} can be obtained using the fact that for any such  $\varepsilon$-projector $\mathcal{A}$ we must have $\varepsilon \le \Delta$, where $\Delta$, defined in \eqref{equ:eigen-dist}, is the minimum spacing between the distinct eigenvalues of $A$. Let us record these observations for the sake of referencing.

\begin{corollary}
    \label{cor:r-eq-1}
    Let $\Gamma = (V, E)$ be a graph with $N$ vertices and adjacency matrix $A$. Given an $\varepsilon$-projector $\mathcal{A} = \{B\}$ for $A$, we have the following:
    \begin{enumerate}[label = (\alph*), leftmargin = *, itemsep = 0mm, topsep = 2mm]
    \item \label{cor:r-eq-1-exp} If the unitary $e^{iB}$ can be applied in $\poly(\log N)$ operations, then there is a quantum algorithm that can sample from the limiting distribution of the walk $W(t) = e^{iAt}$ in $O(\varepsilon^{-1} \poly(\log N))$ operations.
    \item \label{cor:r-eq-1-poly} If the unitary $e^{iBt}$ can be applied in $O((\log t) \poly(\log N))$ operations, then there is a quantum algorithm that can sample from the limiting distribution of the walk $W(t) = e^{iAt}$ in $O(\log(\varepsilon^{-1}) \poly(\log N))$ operations.
    \end{enumerate}
    Moreover, if $\Delta$ is known or can be approximated efficiently, then the above running times can be improved to $O(\Delta^{-1} \poly(\log N))$ in case \ref{cor:r-eq-1-exp} and $O(\log(\Delta^{-1}) \poly(\log N))$ in case \ref{cor:r-eq-1-poly}.
\end{corollary}

In the black-box setting, we are usually given access to the adjacency matrix $A$ of $\Gamma$ such that we can apply the unitary $e^{iA}$ in time $\poly(\log N)$. In this setting, we can just take the $\varepsilon$-projector $\mathcal{A} = \{A\}$ for a small enough $\varepsilon$. This makes the complexity of sampling from the limiting distribution of the walk $W(t)$ fundamentally dependent on $\Delta$. If we use the naive way of running $W(t)$ for a large $t$ and measuring the resulting state, the bound \eqref{equ:mixing-bound} suggests that we should take $T$ proportional to $(2 \ln M + 2) / \delta\Delta$ to be within distance $\delta$ of the limiting distribution. In comparison, Corollary \ref{cor:r-eq-1} says we only need to run $W(t)$ for $t \approx 1 / \Delta$ (and perform some other negligible operations) to sample exactly from the limiting distribution. 

In the non-black-box setting, we have the opportunity to exploit some extra information on $A$ to find $\varepsilon$-projectors that allow us to efficiently sample from the limiting distribution of $W(t)$. In the following sections, we give examples of graphs for which we can find such $\varepsilon$-projectors.


\section{Quasi-Abelian Graphs}
\label{sec:quasi-abelian}

As a first application of our sampling algorithm we consider a class of graphs, called quasi-abelian graphs, in this section. Quasi-abelian graphs, as we will see, are a potential source of concrete examples for which Algorithm \ref{alg:sampling} runs in polynomial time. In the following, we first review some general properties of quasi-abelian graphs, and then look more closely at a specific example called Winnie Li graphs. 

Let $G$ be a finite group of size $N$, and let the subset $S \subseteq G$ be such that $1 \notin S$. The Cayley digraph $\Gamma = \Gamma(G, S)$ of the pair $(G, S)$ is a directed graph in which the vertex set is the set of elements in $G$ and the edge set is $\{(a, as) : a \in G, s \in S \}$. If $S$ is symmetric, i.e., $s \in S$ if and only if $s^{-1} \in S$, then $\Gamma$ is an undirected graph called the Cayley graph. In this paper, we always assume that $S$ generates the entire group $G$, which means that $\Gamma$ is connected. Denote by $f_S: G \rightarrow \{0, 1\}$ the characteristic function of $S$ defined by $f_S(a) = 1$ if $a \in S$ and $f_S(a) = 0$ if $a \notin S$. We will also denote by $A(\Gamma)$ the adjacency matrix of $\Gamma$. 

\begin{definition}
    A quasi-abelian graph is a Cayley graph $\Gamma(G, S)$ in which $S$ is the union of some conjugacy classes of $G$.
\end{definition}

A class function is a function $f: G \rightarrow \C$ that is constant on conjugacy classes of $G$. It is not hard to show that the Fourier transform of a class function $f$ is a diagonal matrix, e.g., see \cite[Chapter 2]{diaconis1988group}. From the definition of quasi-abelian graphs we see that $f_S$ is always a class function. For any irreducible representation $\varrho \in \widehat{G}$ we obtain
\begin{equation}
    \label{equ:quasi-g-eigen}
    \widehat{f}_S(\varrho) = \frac{1}{\sqrt{d_\varrho \abs{G}}} \sum_{s \in S} \chi_\varrho(s) \mathds{1}_{d_\varrho}
\end{equation}
where $d_\varrho$ is the dimension of $\varrho$ and $\chi_\varrho$ is the character of $\varrho$. The following proposition is partially proved in \cite{diaconis1988group} and \cite{rockmore2002fast} with different notations. Here, we give a short proof consistent with our notations. 

\begin{theorem}
    \label{thm:quasi-diag}
    Let $\Gamma(G, S)$ be a quasi-abelian graph on a finite group $G$. Denote by $F_G$ the Fourier transform over $G$. Then
    \begin{enumerate}[label = (\alph*), align = left, leftmargin = *, itemsep = 0mm, topsep = 2mm]
    \item \label{thm:quasi-diag-m} The adjacency matrix $A(\Gamma)$ is diagonalized by $F_G$,
    \item \label{thm:quasi-diag-vec} The eigenvectors of $A(\Gamma)$ are given by $F_G^* \ket{\varrho, j, k}$ for $\varrho \in \widehat{G}$ and $1 \le j, k \le d_\varrho$,
    \item \label{thm:quasi-diag-eig} The eigenvalue corresponding to the eigenvector $F_G^* \ket{\varrho, j, k}$ is given by
        \[ \lambda_\varrho = \frac{1}{d_\varrho} \sum_{s \in S} \chi_\varrho(s). \]
    \end{enumerate}
\end{theorem}

From part \ref{thm:quasi-diag-eig} of the theorem we see that the eigenvalues of $A(\Gamma)$ are only determined by the irreducible representations of $G$ and the set $S$. Each irreducible representation $\varrho$ corresponds to $d_\varrho^2$ eigenvectors. This means an eigenvalue $\lambda_\varrho$ has multiplicity at least $d_\varrho^2$. If $G$ is abelian, we always have $d_\varrho = 1$, but that does not mean the eigenvalues $\lambda_\varrho$ are distinct for different $\varrho \in \widehat{G}$.

\begin{proof}[Proof of Theorem \ref{thm:quasi-diag}]
    Let $\rho_\mathrm{reg}$ be the regular representation of $G$. Then we can write
    \[ A(\Gamma) = \sum_{s \in S} \rho_\mathrm{reg}(s) = \sum_{a \in G} f_S(a) \rho_\mathrm{reg}(a). \]
    The Fourier transform $F_G$ decomposes $\rho_\mathrm{reg}$ as
    \[ F_G \rho_\mathrm{reg}(a) F_G^* = \bigoplus_{\varrho \in \widehat{G}} (\varrho(a) \otimes \mathds{1}_{d_\varrho}). \]
    It follows from this decomposition that
    \begin{align*}
        F_G A(\Gamma) F_G^*
        & = \sum_{a \in G} \bigoplus_{\varrho \in \widehat{G}} (f_S(a) \varrho(a) \otimes \mathds{1}_{d_\varrho}) \\
        & = \bigoplus_{\varrho \in \widehat{G}} \sqrt{\frac{\abs{G}}{d_\varrho}} (\widehat{f}_S(\varrho) \otimes \mathds{1}_{d_\varrho}) \\
        & = \bigoplus_{\varrho \in \widehat{G}} (\lambda_\varrho \mathds{1}_{d_\varrho} \otimes \mathds{1}_{d_\varrho}),
    \end{align*}
    where the last equality follows from the fact that $f_S$ is a class function. This proves \ref{thm:quasi-diag-m}. Parts \ref{thm:quasi-diag-vec} and \ref{thm:quasi-diag-eig} follow from the definition of the quantum Fourier transform \eqref{equ:qft-g} and the identity \eqref{equ:quasi-g-eigen} for the characteristic function $f_S$.
\end{proof}

According to Theorem \ref{thm:quasi-diag}, for a quasi-abelian graph $\Gamma(G, S)$ we can write the adjacency matrix $A = A(\Gamma)$ as
\begin{equation}
    \label{equ:quasi-adj}
    A = \sum_{\varrho \in \widehat{G}} \; \sum_{1 \le j, k \le d_\varrho} \lambda_\varrho F_G^* \ket{\varrho, j, k}\bra{\varrho, j, k} F_G.
\end{equation}
Therefore,
\begin{equation}
    \label{equ:quasi-walk}
    e^{iAt} = \sum_{\varrho \in \widehat{G}} \sum_{1 \le j, k \le d_\varrho} e^{i\lambda_\varrho t} F_G^* \ket{\varrho, j, k}\bra{\varrho, j, k} F_G.
\end{equation}
The expansion \eqref{equ:quasi-walk} suggest that we could perform the walk $e^{iAt}$ using the following three steps:
\begin{enumerate}[leftmargin = *, itemsep = 0mm]
\item Apply the quantum Fourier transform $F_G$,
\item Apply the phase operator $U_\varrho: \ket{\varrho, j, k} \mapsto e^{i \lambda_\varrho t} \ket{\varrho, j, k}$,
\item Apply the inverse quantum Fourier transform $F_G^*$.
\end{enumerate}
Of course, this is only efficient if both the Fourier transform $F_G$ and the phase operator $U_\varrho$ can be applied efficiently. Suppose we can apply $F_G$ efficiently for a given group $G$. Then a general strategy for constructing $\varepsilon$-projectors for $A$ is as follows. For any function $g: \widehat{G} \rightarrow \C$ the operator
\[ A_g = \sum_{\varrho \in \widehat{G}} \sum_{1 \le j, k \le d_\varrho} g(\varrho) F_G^* \ket{\varrho, j, k}\bra{\varrho, j, k} F_G  \]
commutes with $A$. Suppose that $g$ satisfies the following condition:
\begin{equation}
    \label{equ:same-egn-spc}
    g(\varrho_1) = g(\varrho_2) \text{ if and only if } \lambda_{\varrho_1} = \lambda_{\varrho_2} \text{ for all } \varrho_1, \varrho_2 \in \widehat{G}
\end{equation}
Then $A_g$ has the same eigenspaces as $A$. If we can efficiently approximate $g$ with exponential accuracy, then $\mathcal{A} = \{A_g\}$ is an $\varepsilon$-projector that satisfies the conditions of Corollary \ref{cor:r-eq-1} \ref{cor:r-eq-1-poly}. If we can only approximate $g$ with polynomial accuracy, then we might need many more of these functions $g$ that satisfy \eqref{equ:same-egn-spc}. Ideally, we need to find $g_1, g_2, \dots, g_r: \widehat{G} \rightarrow \C$, with $r = \poly(\log N)$, such that $\mathcal{A} = \{ A_{g_j} \}_{1 \le j \le r}$ is an $\varepsilon$-projector for $A$ for some $\varepsilon = 1 / \poly(\log N)$. In this case, $\mathcal{A}$ satisfies the conditions of Theorem \ref{thm:main}, and we can efficiently sample from the limiting distribution of the walk $W(t) = e^{iAt}$ using $\mathcal{A}$.

\subsection{Winnie Li graphs}

Let $\F_p$ be a finite field of characteristic $p \ge 3$. For an extension $F / \F_p$ of degree $n$, the norm map $\norm_{F / \F_p}: F \rightarrow \F_p$ is defined by $\norm_{F / \F_p}(a) = a^{(p^n - 1) / (p - 1)}$, which is a homomorphism between the multiplicative groups $F^\times$ and $\F_p^\times$. Let $S = \ker(N_{F / \F_p})$, i.e., the set of elements of $F$ of norm $1$. The Winnie Li graph of $F$ over $\F_p$ is the Cayley digraph $\Gamma(F, S)$ where the vertex set is $F$ and the edge set is $\{ (a, a + s) : a \in F, s \in S \}$.

Before we get into the specifics of the structures of these graphs, recall the quantum Fourier transform over $F$. The set of additive characters of $F$ is given by
\[ \chi_{F, a}(x) = \omega_p^{\tr_{F / \F_p}(ax)}, \quad a \in F, \]
where $\tr_{F / \F_p}(x) = x + x^p + \cdots + x^{p^{n - 1}}$ is the trace map from $F$ to $\F_p$. The quantum Fourier transform of the basis element $\ket{a}$, where $a \in F$, is given by
\[ \ket{\hat{a}} = \frac{1}{\sqrt{\abs{F}}} \sum_{x \in F} \chi_{F, a}(x) \ket{x} \]

For even $n$, the graph $\Gamma$ is undirected, since $\norm_{F / \F_p}(-a) = \norm_{F / \F_p}(a)$. Let us look at the simple case where $n = 2$. The extension $F / \F_p$ is a quadratic extension, and if we assume that $-1$ is a quadratic nonresidue in $\F_p$, we can take $F = \F_p(i)$ where $i^2 = -1$. The elements of $\F_p(i)$ can be written in the form $x + iy$ for $x, y \in \F_p$, so the norm map takes the simple form
\[ \norm_{\F_q(i) / \F_p}(x + iy) = x^2 + y^2. \]
Therefore, the set $S$ of elements of norm $1$ is the set of $\F_p$-points of the circle $x^2 + y^2 = 1$. The Winnie Li graph $\Gamma(\F_p(i), S)$ is a $(p + 1)$-regular graph of $p^2$ vertices. It follows from \eqref{equ:quasi-adj} that the adjacency matrix of $\Gamma$ can be written as
\[ A(\Gamma) = \sum_{a \in F} \lambda_a \ket{\hat{a}}\bra{\hat{a}}, \]
For an element $a = u + iv \in \F_p(i)$, Theorem \ref{thm:quasi-diag} \ref{thm:quasi-diag-eig} gives
\begin{equation}
    \label{equ:winnie-eigen-2}
    \lambda_a = \sum_{\norm_{F / \F_p}(b) = 1} \chi_{F, a}(b) = \sum_{x^2 + y^2 = 1} \omega_p^{2(ux - vy)} =
    \begin{cases}
        q + 1 & \text{if } a = 0, \\
        -K(1, u^2 + v^2) & \text{if } a \ne 0.
    \end{cases}
\end{equation}
Here, $K(a, b)$ is the Kloosterman sum with parameters $a, b$, which we will briefly talk about next.

\paragraph{Kloosterman sums.}
For $a, b \in \F_p$, the exponential sum
\begin{equation}
    \label{equ:Klooster-sum}
    K(a, b) = \sum_{x \in \F_p^\times} \omega_p^{ax + bx^{-1}}
\end{equation}
is called the Kloosterman sum with respect to $a, b$. The last equality in \eqref{equ:winnie-eigen-2} follows from the definition \eqref{equ:Klooster-sum}. Since $\overline{K(a, b)} = K(a, b)$, these sum are real numbers. A well known result of Weil \cite{weil1948some} gives the bound $\abs{K(a, b)} \le 2\sqrt{p}$. When $p$ is large, estimating $K(a, b)$ is an open problem; there are no known classical or quantum algorithms that can efficiently estimate $K(a, b)$. For a multiplicative character $\chi$ of $\F_p^\times$, the $\chi$-twisted Kloosterman sum is defined by
\begin{equation}
    \label{equ:Klooster-sum-t}
    K(\chi, a, b) = \sum_{x \in \F_p^\times} \chi(x) \omega_p^{ax + bx^{-1}}.
\end{equation}
Interestingly, when $\chi$ is the quadratic character of $\F_p^\times$, the sum \eqref{equ:Klooster-sum-t} has a closed form, and is easy to estimate \cite{salie1932kloostermanschen, carlitz1953weighted}.

\paragraph{Euclidean graphs.}
The construction of the Winnie Li graph $\Gamma(\F_p(i), S)$ can be directly generalized to obtain the so called Euclidean graphs \cite{medrano1996finite}. Let $n > 0$ be an integer and $b \in \F_p$. Define the quadratic form $Q(x) = x_1^2 + x_2^2 + \cdots +x_n^2$ over $\F_p$. An Euclidean graph for $n$ and $b$ is the Cayley graph $\Gamma(\F_p^n, S_b)$ where $S_b$ is the set of solutions of $Q(x) = b$ in $\F_p^n$. For simplicity, assume $b = 1$. The adjacency matrix of $\Gamma$ is
\[ A(\Gamma) = \sum_{a \in \F_p^n} \lambda_a \ket{\hat{a}}\bra{\hat{a}} \]
where
\[ \ket{\hat{a}} = \frac{1}{\sqrt{p^n}} \sum_{x \in \F_p^n} \omega_p^{\lrang{a, x}} \ket{x} \]
is the quantum Fourier transform of $\ket{a}$, and $\lrang{a, x} = a_1x_1 + \cdots + a_nx_n$ for $a = (a_1, \dots, a_n)$ and $x = (x_1, \dots, x_n)$. The eigenvalues $\lambda_a$ are \cite{medrano1996finite}
\begin{equation}
    \label{equ:Euc-eigen}
    \lambda_a = \sum_{x \in S_1} \omega_p^{\lrang{a, x}} =
    \begin{cases}
        \abs{S_1} & \text{if } a = 0, \\
        \frac{G_1^n}{p} K\Big(\chi^n, 1, \frac{Q(a)}{4}\Big) & \text{if } a \ne 0,
    \end{cases}
\end{equation}
where $\chi$ is the quadratic character of $\F_p^\times$, and where
\[G_1 = 
\begin{cases}
    \sqrt{p} & \text{if } p \equiv 1 \bmod 4, \\
    i\sqrt{p} & \text{if } p \equiv 3 \bmod 4.
\end{cases}
\]

When $n$ is odd, it is easy to approximate the eigenvalues $\lambda_a$ with exponential accuracy, since it is easy to approximate \eqref{equ:Klooster-sum-t} when $\chi$ is the quadratic character. So, it is easy to perform the walk
\[ e^{iA(\Gamma)t} = \sum_{a \in \F_p^n} e^{i\lambda_a t} \ket{\hat{a}}\bra{\hat{a}} \]
for exponentially large $t$. Therefore, we can easily sample from the limiting distribution of $W(t)$ by running the walk for random values of $t \in [0, T]$ for a large $T$. However, when $n$ is even, we do not know how to perform $W(t)$ for large $t$. In fact, there is no known way to even perform $W(1) = e^{iA(\Gamma)}$ efficiently. We now use the general strategy introduced at the beginning of Section \ref{sec:quasi-abelian} to efficiently sample from the limiting distribution of $W(t)$ when $n$ is even. Note that since $G = \F_p^n$ is an abelian group, we have $\widehat{G} \cong G$. Therefore, we need to find functions $g: G \rightarrow \C$ that satisfy the condition in \eqref{equ:same-egn-spc}.

\begin{proposition}
    \label{prop:Euc-e-proj}
    Define the function $g: G \rightarrow \C$ as $g(x) = (Q(x) \bmod p) / p$. Then the operator
    \[ A_g = \sum_{a \in \F_p^n} g(a) \ket{\hat{a}}\bra{\hat{a}} \]
    is an $\varepsilon$-projector for $A(\Gamma)$ for $\varepsilon = 1 / p$.
\end{proposition}
\begin{proof}
    We first need to show that $A_g$ has the same eigenspaces as $A(\Gamma)$. It is known that the Kloosterman sums $K(1, b)$, $b \in \F_p$, are distinct \cite{fisher1990distinctness}. Therefore, according to \eqref{equ:Euc-eigen}, two eigenvalues $\lambda_a$ and $\lambda_b$ of $A(\Gamma)$ are the same if and only if $Q(a) = Q(b)$. But we also trivially have $g(a) = g(b)$ if and only if $Q(a) = Q(b)$. Note that the eigenspace corresponding to $\lambda_a$ is the set of all vectors $\ket{\hat{x}}$ such that $Q(x) = Q(a)$. 

    Now, for $\lambda_a \ne \lambda_b$ we have $Q(a) \ne Q(b)$, hence
    \[ \abs{g(a) - g(b)} = \frac{1}{p} \abs[\big]{(Q(a) - Q(b)) \bmod p} \ge \frac{1}{p}, \]
    and since we can efficiently compute $g$ with exponential accuracy, it follows that $A_g$ is a $\frac{1}{p}$-projector for $A(\Gamma)$.
\end{proof}

With $g$ as in Proposition \ref{prop:Euc-e-proj}, we can perform the walk
\[ e^{iA_gt} = \sum_{a \in \F_p^n} e^{ig(a)t} \ket{\hat{a}}\bra{\hat{a}} \]
in $O((\log t) \poly(\log N))$ operations. Corollary \ref{cor:r-eq-1} \ref{cor:r-eq-1-poly} and Proposition \ref{prop:Euc-e-proj} now give

\begin{theorem}
    \label{thm:Euc-sample}
    Let $\Gamma(\F_p^n, S)$ be an Euclidean graph with adjacency matrix $A(\Gamma)$. For a given initial state $\ket{\psi_0}$ we have
    \begin{enumerate}[label = (\alph*), leftmargin = *, itemsep = 0mm, topsep = 2mm]
        \item The limiting distribution of the walk $W(t) = e^{iA(\Gamma)t}$ on $\Gamma$ is given by
            \[ P_\infty(v \vert \psi_0) = \sum_{j = 0}^{p - 1} \abs[\Big]{\sum_{k \in S_j} \braket{v}{\hat{k}} \braket{\hat{k}}{\psi_0}}^2. \]
        \item There is a quantum algorithm that can sample from $P_\infty(\cdot \vert \psi_0)$ in $\poly(n\log p)$ operations. 
    \end{enumerate}
\end{theorem}


\section{Isogeny Graphs}
\label{sec:isogeny}

As another application of Algorithm \ref{alg:sampling} we consider a class of graphs, called isogeny graphs, in this section. Isogeny graphs have attracted much attention in last two decade mainly because of their applications in cryptography \cite{de2017mathematics, de2018exploring}. We will see in the following how the beautiful theory of Hecke algebras give us natural $\varepsilon$-projectors that make it possible to efficiently sample from the limiting distributions of quantum walks on these graphs.

The $\varepsilon$-projectors considered in this section, were first implicitly used by Kane, Sharif and Silverberg \cite{kane2021quantum}. They considered the same set of operators we use here but in the context of quaternion algebras. In here, we consider the space of supersingular elliptic curves over the finite fields $\F_{p^2}$, whereas in \cite{kane2021quantum} they considered the space of ideals in an ideal class of the quaternion algebra $B_{p, \infty}$. These two spaces are mathematically essentially the same, in a precise sense, but they are vastly different from a cryptographic perspective. In particular, the isogeny problem is easy in the latter space, but believed to be hard in the former \cite{kohel2014quaternion}. Working with the $\varepsilon$-projectors in the space of elliptic curves yields a potential solution to the problem of generating an honest curve, as we will explain in Section \ref{sec:honest-vurve}.

Let $\F_q = \F_{p^2}$ be a finite field of characteristic $p \ge 5$. An elliptic curve $E$ over $\F_q$ is a projective smooth curve of genus one. For the definition of these terms and an extensive introduction to elliptic curve see \cite{husemoller2013elliptic}. The affine version of $E$ is usually written as the cubic $y^2 = x^3 + ax + b$, $a, b \in \F_q$, known as the Weierstrass equation of $E$. The set of points on $E$ in any extension of $\F_p$ form an abelian group. An elliptic curve $E$ over $\F_q$ is called supersingular if it has no nontrivial points of order $p$. It can be shown that any supersingular elliptic curves over the algebraic closure $\overline{\F}_p$ can be defined over $\F_q$, or, more precisely, is isomorphic to a curve over $\F_q$. Therefore, we always assume that any supersingular elliptic curve $E$ has its coefficients $a, b$ in $\F_q$.

An isogeny $\phi: E_1 \rightarrow E_2$ between two elliptic curves $E_1$ and $E_2$ is a rational function that is also a homomorphism of groups of points on $E_1$ and $E_2$. An isogeny $\phi$ induces an embedding of function fields $\phi^*: K(E_2) \rightarrow K(E_1)$ defined by $\phi^*(f) = f \circ \phi$. The degree of $\phi$ is the degree of the extension $K(E_1) / \phi^*K(E_2)$. An isogeny of degree $\ell$ is called an $\ell$-isogeny. For any isogeny $\phi$ there is a unique isogeny $\widehat{\phi}: E_2 \rightarrow E_1$ called the dual of $\phi$. Let $\ell$ be a prime different from the characteristic $p$. Define a graph $G_\ell$ with vertices the set of all $\overline{\F}_p$-isomorphism classes of supersingular elliptic curves, and edges the set of $\ell$-isogenies between the curves. $G_\ell$ is called the supersingular $\ell$-isogeny graph. Since the dual of an $\ell$-isogeny is again an $\ell$-isogeny but in the opposite direction, we usually consider $G_\ell$ as an undirected graph. For simplicity, assume that $p \equiv 1 \bmod 12$. Then $G_\ell$ is an $(\ell + 1)$-regular graph with $N = \lfloor p / 12 \rfloor$ vertices and no self loops. The adjacency matrix of $G_\ell$ is a symmetric matrix denoted by $T_\ell$ and is called the Hecke operator.

\subsection{Simulating the Hecke operators}

Let $S$ be the set of vertices of $G_\ell$, i.e., the set of isomorphism classes of supersingular elliptic curves in characteristic $p$. The Hecke operator $T_\ell$ acts on the formal abelian group
\[ M = \bigoplus_{E \in S} \Z E \]
by sending each curve $E$ to a sum of its neighbours in $G_\ell$. In the quantum setting, we consider the action of $T_\ell$ on the complex Euclidean space $\X = M \otimes_{\Z} \C$ with the basis $\{ \ket{E} \}_{E \in S}$. The operators $T_\ell$, for different values of $\ell$, are closely related to the Hecke operators acting on the space of modular forms \cite{eichler1973basis, kohel1999computing, bilyk2022cryptanalysis}, so the terminology we use here is mostly adopted from the theory of modular forms. For example, the trivial eigenvector
\[ \ket{\mathcal{E}} = \frac{1}{\sqrt{N}} \sum_{E \in S} \ket{E} \]
of $T_\ell$ is called the Eisenstein eigenform and corresponds to the eigenvalue $\lambda_\mathcal{E} = \ell + 1$. Deligne's proof of the Riemann
hypothesis for function fields \cite{deligne1974conjecture, katz1976overview} implies that the nontrivial eigenvalues of $T_\ell$ are contained in the interval $[-2\sqrt{\ell}, 2\sqrt{\ell}]$. We make the heuristic assumption that the eigenvalues of $T_\ell$ are distinct. This assumption is in fact a consequence of (the well known) Maeda's conjecture \cite{hida1997non} which states that the characteristic polynomial of $T_\ell$ is irreducible over $\Q$. We refer the reader to \cite{martin2021average, tsaknias2014possible, ghitza2012experimental}, and the references therein, for results on the computational verification of Maeda's conjecture. 

An $\ell$-isogeny can be computed in $O(\ell)$ operations over $\F_q$ using the V\'elu formulas \cite{velu1971isogenies}. When $\ell$ is small, i.e., $\ell = \poly(\log N)$, it is easy to compute the list of all the neighbours of a given curve $E$ in $G_\ell$. In particular, we can efficiently implement the isometry
\[ T: \ket{E} \rightarrow \frac{1}{\sqrt{\ell + 1}} \sum_{j = 0}^{\ell} \ket{E}\ket{E_j} \]
from $\X$ to $\X \otimes \X$, where $E_1, \dots, E_\ell$ are the $\ell + 1$ neighbours of $E$. Therefore, we can use the existing Hamiltonian simulation techniques \cite{childs2010relationship, watrous2001quantum} to efficiently approximate the unitary $e^{iT_\ell}$.

\begin{proposition}
    \label{prp:Hecke-sim}
    For any prime $\ell = \poly(\log N)$, the walk $W(t) = e^{iT_\ell t}$ on the $\ell$-isogeny graph $G_\ell$ can be performed in $O(t \poly(\log N))$ operations.
\end{proposition}

\subsection{Distribution of eigenvalues}

For different primes $\ell$ we have different isogeny graphs $G_\ell$ over $\F_q$; the set of vertices is always the same but the set of edges change with $\ell$. Therefore, we get different Hecke operators $T_\ell$ acting on the same space $\X$. More generally, Hecke operators can be defined for any integer $n > 0$. The operator $T_n$ represent the adjacency matrix of the $n$-isogeny graph $G_n$, although for non-prime $n$ one needs to be more careful about some definitions. The algebra $\mathbb{T}_\Z = \Z[\{T_n\}_{n \in \Z}]$, generated by all Hecke operators acting on $\X$, is called the Hecke algebra. Let $\mathbb{T} = \mathbb{T}_\Z \otimes_{\Z} \C$ be the Hecke algebra over $\C$. It can be shown that $\mathbb{T}$ is a commutative ring \cite[Chapter 41]{voight2021quaternion}. In particular, for any $m, n$, the operators $T_n$ and $T_m$ commute. This means all Hecke operators are simultaneously diagonalizable. 

It was proved by Serre \cite{serre1997repartition} that for large $p$, the eigenvalues of the normalized Hecke operator $T_\ell / \sqrt{\ell}$ are equidistributed in $[-2, 2]$ with respect to the measure
\[ \mu_\ell = \frac{\ell + 1}{\pi} \frac{(1 - x^2 / 4)^{1 / 2} dx}{(\ell^{1 / 2} + \ell^{-1 / 2})^2 - x^2}. \] 
Let $\ket{\phi_1}, \ket{\phi_2}, \dots, \ket{\phi_N} \in \X$ be a simultaneous basis for all Hecke operators, and let $\ell_1, \ell_2, \dots, \ell_r$ be a set of distinct primes each bounded by $\poly(\log N)$. For all $1 \le k \le r$ and $1 \le j \le N$ we have
\[\frac{1}{\sqrt{\ell_k}} T_{\ell_k} \ket{\phi_j} = \lambda_{j, k} \ket{\phi_j}\]
for some $\lambda_{j, k} \in [-2, 2]$. Define $\bm{\lambda}_j = (\lambda_{j, 1}, \lambda_{j, 2}, \dots, \lambda_{j, r})$. It was also proved in \cite{serre1997repartition} that for large $p$, the vectors $\bm{\lambda}_j$ are equidistributed in $[-2, 2]^r$ with respect to the product measure $\bm{\mu} = \prod_{k = 1}^r \mu_{\ell_k}$. This means that asymptotically we can treat the $\bm{\lambda}_j$ as independent samples from the distribution given by the measure $\bm{\mu}$. 

For large $N$, $\mu_\ell$ approaches the measure $\mu = \frac{1}{2\pi}\sqrt{4 - x^2} dx$, the semicircle distribution on $[-2, 2]$. So, when the characteristic $p$ is large enough, it is natural to assume that the vectors $\bm{\lambda}_j$ are independent random samples from $\mu^r$. We now prove that for $r \in O(\log N)$ and $\varepsilon = 1 / \sqrt{\log N}$, the set of Hecke operators $\mathcal{T} = \{T_{\ell_k} / \sqrt{\ell_k}\}_{1 \le k \le r}$ is an $\varepsilon$-projector for any Hecke operator $T_\ell$ with $\ell$ a prime number. The following proves the $\varepsilon$-separatedness of $\mathcal{T}$.

\begin{lemma}
    \label{lem:isog-e-sep}
    Let $r \ge  32\log N$, let $\ell_1, \ell_2, \dots, \ell_r$ be a set of distinct primes each bounded by $\poly(\log N)$, and let $\varepsilon = 1 / \sqrt{\log N}$. Then the set of operators $\mathcal{T} = \{ T_{\ell_k} / \sqrt{\ell_k} \}_{1 \le k \le r}$ is $\varepsilon$-separated with overwhelming probability.
\end{lemma}
\begin{proof}
    Let $Z = (X - Y)^2$ where $X$ and $Y$ are distributed according to $\mu$. We have $\mathrm{Var}(X) = \mathrm{Var}(Y) = 1$ and $\E[X] = \E[Y] = 0$. So
    \[ \E[Z] = \mathrm{Var}(X - Y) + \E[X - Y]^2 = 2. \]
    Now define $Z_k = (\lambda_{j, k} - \lambda_{l, k})^2$ for $k = 1, \dots, r$, and $W = Z_1 + Z_2 + \cdots + Z_r$. Then $\E[W] = 2r$. We have
    \begin{align}
        \Pr[\opnorm{\bm{\lambda}_j - \bm{\lambda}_l}_2 \le \varepsilon]
        & = \Pr[W \le \varepsilon^2] \nonumber \\
        & = \Pr[W - 2r \le \varepsilon^2 - 2r] \nonumber \\
        & \le \exp\Big( \frac{-2(2r - \varepsilon^2)^2}{256 r} \Big) \label{equ:hoeffding-ineq} \\
        & \le \frac{1}{N^{1.4}} \nonumber
    \end{align}
    where the inequality \eqref{equ:hoeffding-ineq} is the Hoeffding's inequality \cite{hoeffding1994probability}. 
\end{proof}

Lemma \ref{lem:isog-e-sep} and Theorem \ref{thm:main} now give

\begin{theorem}
    \label{thm:isog-sample}
    Let $\ell \ne p$ be prime, and let $G_\ell$ be the $\ell$-isogeny graph with vertices the set of supersingular elliptic curves over $\F_{p^2}$. For a given initial state $\ket{\psi_0}$ we have
    \begin{enumerate}[label = (\alph*), leftmargin = *, itemsep = 0mm, topsep = 2mm]
    \item The limiting distribution of the walk $W(t) = e^{iT_\ell t}$ on $G_\ell$ is given by
        \begin{equation}
            \label{equ:isogeny-lim-dist}
            P_\infty(E \vert \psi_0) = \sum_{j = 1}^N \abs{\braket{E}{\phi_j} \braket{\phi_j}{\psi_0}}^2
        \end{equation}
    \item There is a quantum algorithm that can sample from $P_\infty(\cdot \vert \psi_0)$ in $\poly(\log p)$ operations.
    \end{enumerate}
\end{theorem}

\subsection{Honest hard curves}
\label{sec:honest-vurve}

A hard problem in isogeny based cryptography is to compute the endomorphism ring $\groupofend(E)$ of a given supersingular elliptic $E$. The majority of other computational assumptions can be reduced to the endomorphism ring problem \cite{wesolowski2022orientations}. Informally, a \textit{hard curve} is a random curve on $G_\ell$ with an unknown endomorphism ring.  
To classically generate a random curve on $G_\ell$, one can do the following:
\begin{enumerate}[itemsep = 0mm, parsep = 1mm, topsep = 2mm]
\item Start with a known curve $E_0$ on $G_\ell$,
\item Take a random walk of length at least $2\log p$, and
\item Return the final curve
\end{enumerate}
For small $\ell$, taking classical random walks on $G_\ell$ can be done very efficiently \cite{doliskani2017faster}. Also, it is well known that $G_\ell$ has a small diameter, and taking a walk of length $\approx 2\log p$ is enough to get close to the uniform distribution on $G_\ell$.

At the first glance, it seems that the above classical random walk on $G_\ell$ can be used to efficiently generate a hard curve. However, there is an issue with this approach that seems to be unavoidable: the random walk explicitly generates a path on $G_\ell$. This path can be used to compute the endomorphism ring of the returned curve. More precisely, if the endomorphism ring of the initial curve $E_0$ is known and we are given a path $\phi: E_0 \rightarrow E$ (which is an isogeny) then we can use $\phi$ to compute $\groupofend(E)$. Any such path $\phi$ is called a backdoor for $E$. A hard curve without a backdoor is called an honest hard curve. 

There is no known efficient classical solution for generating an honest hard curve. A potential solution using quantum walks was first discussed in \cite{booher2022failing}. The idea is to sample from the limiting distribution of the walk $W(t) = e^{iT_\ell t}$ on $G_\ell$. According to Theorem \ref{thm:isog-sample}, this can be done efficiently, and in contrast to the classical walk, the quantum walk does not generate any path on $G_\ell$. Two important questions about the quantum walk solution in \cite{booher2022failing} remained unresolved. We briefly address those questions here.

The first question is weather the distribution \eqref{equ:isogeny-lim-dist} is close to uniform. Intuitively, there is no reason to believe that the eigenvectors $\ket{\phi_j}$ are localized\footnote{A unit vector is localized if the mass of the vector is concentrated on a small subset of entries.}. In particular, the action of the Hecke operator $T_\ell$ on $\X$ implies that each entry of $\ket{\phi_j}$ is an average of a set of $\ell + 1$ other entries. Since $\ket{\phi_j}$ is an eigenvector for all Hecke operators, these averages involve all the other entries for large enough $\ell$. Therefore, one would expect that the entries of $\ket{\phi_j}$ are not too small or too large, and that the distribution $P_\infty(\cdot \vert \psi_0)$ is not concentrated on a small subset of vertices. In fact, it is not hard to show, using general techniques in the theory of Markov chains, that $P_\infty(E_1 \vert E_2) \ge N^{-2}$ for every two curves $E_1, E_2$ \cite{richter2007almost}. Notwithstanding, a rigorous proof that $P_\infty(\cdot \vert \psi_0)$ is close to the uniform distribution does not seem straightforward. Maybe there are ways to analyze the amplitudes of the vectors $\ket{\phi_j}$ through their close connection to complex modular forms, but we are not aware of any work regarding this in the literature. Instead, this question can be approached algorithmically using the double-loop technique of \cite{richter2007almost}. The double-loop algorithm for our case is:
\begin{enumerate}[leftmargin = *, parsep = 0mm, itemsep = 0mm]
\item Set $E$ to a known curve $E_0$ on $G_\ell$
\item Repeat for $k$ times
    \begin{enumerate}[leftmargin = *, topsep = 0mm, parsep = 0mm, itemsep = 0mm]
    \item Run Algorithm \ref{alg:sampling} with initial state $\ket{E}$ to obtain a curve $E_1$
    \item Set $E \leftarrow E_1$
    \end{enumerate}
\item Return $E$    
\end{enumerate}
Assuming that the maximum column distance
\[ \alpha = \max_{E_1, E_2 \in S} \opnorm{P_\infty(\cdot \vert E_1) - P_\infty(\cdot \vert E_2)}_1 \]
of $P_\infty$ is bounded by a constant smaller than $1$, we need to only set $k = \lceil \log_{1 / \alpha} \delta^{-1} \rceil$ for the distribution of the final curve $E$ to be within distance $\delta$ of the uniform distribution. Again, the assumption that $\alpha$ is always bounded by a constant $c < 1$ is not rigorously proved, but it is more plausible than the assumption that $P_\infty(\cdot \vert \phi_0)$ is close to uniform.

The second question is only specific to the sampling algorithm proposed in \cite{booher2022failing}. In that algorithm, after preparing the state \eqref{equ:many-phase-g}, the first register is measured. The measurement outcome is a random vector $\tilde{\bm{\lambda}}_j$, and the post-measurement state is the corresponding eigenstate $\ket{\phi_j}$. The state $\ket{\phi_j}$ is then measured in the vertex basis to obtain a curve $E$. The question is whether the vector $\tilde{\bm{\lambda}}_j$ reveals any information about the endomorphism ring $\groupofend(E)$ of the curve $E$. This situation is entirely avoided in Algorithm \ref{alg:sampling}; we never measure the first register. The post-measurement state of Algorithm \ref{alg:sampling} is a superposition of all the vectors $\tilde{\bm{\lambda}}_j$, $j = 1, \dots, N$, which does not seem to provide any useful information about $\groupofend(E)$.

\begin{remark}
    In a cryptography scenario, if Alice presents Bob with an alleged hard curve $E$, there is no way for Bob to know whether $E$ is an honest hard curve, or that Alice is in possession of a backdoor for $E$. In other words, Bob is unable to determine which algorithm Alice has used to generate $E$. This kind of trust issue is normally solved by a higher level cryptographic construction.
\end{remark}


\newpage
\bibliographystyle{plain}
\bibliography{references}

\newpage
\appendix

\section{Bound on the mixing time}
\label{sec:mixing:proof}

In this section, we bound the mixing time $M_\delta$ defined in \eqref{equ:mixing-time}. Our analysis closely follows that of \cite{aharonov2001quantum}. We have
\begin{align}
    \abs{P_T(v \vert \psi_0) - P_\infty(v \vert \psi_0)}
    & = \abs[\Bigg]{\sum_{k, \ell = 1 : \lambda_k  \ne \lambda_\ell}^N \braket{v}{\phi_k} \braket{\phi_k}{\psi_0} \braket{\psi_0}{\phi_\ell} \braket{\phi_\ell}{v} \frac{e^{i(\lambda_k - \lambda_\ell)T} - 1}{i(\lambda_k - \lambda_\ell)T}} \nonumber \\
    & \le \sum_{k, \ell = 1 : \lambda_k  \ne \lambda_\ell}^N \frac{2}{\abs{\lambda_k - \lambda_\ell}T} \abs{\braket{v}{\phi_k} \braket{\phi_\ell}{v}} \cdot \abs{\braket{\phi_k}{\psi_0} \braket{\psi_0}{\phi_\ell}}. \label{equ:mixing-inner}
\end{align}
Using the inequality $2\abs{ab} \le \abs{a}^2 + \abs{b}^2$ we can bound  \eqref{equ:mixing-inner} by
\[ \sum_{k, \ell = 1 : \lambda_k  \ne \lambda_\ell}^N \frac{1}{\abs{\lambda_k - \lambda_\ell}T} (\abs{\braket{v}{\phi_k}}^2 + \abs{\braket{\phi_\ell}{v}}^2) \cdot \frac{1}{2}(\abs{\braket{\phi_k}{\psi_0}}^2 + \abs{\braket{\psi_0}{\phi_\ell}}^2). \]
Summing over all $v \in V$ we obtain
\begin{align*}
    \opnorm{P_T(\cdot \vert \psi_0) - P_\infty(\cdot \vert \psi_0)}_1
    & \le \sum_{k, \ell = 1 : \lambda_k  \ne \lambda_\ell}^N \frac{\abs{\braket{\phi_k}{\psi_0}}^2 + \abs{\braket{\psi_0}{\phi_\ell}}^2}{\abs{\lambda_k - \lambda_\ell}T} \\
    & = \sum_{h, j : h \ne j}^M \; \sum_{k \in I_h, \ell \in I_j} \frac{\abs{\braket{\phi_k}{\psi_0}}^2 + \abs{\braket{\psi_0}{\phi_\ell}}^2}{\abs{\lambda_k - \lambda_\ell}T}
\end{align*}
Define $\beta_j = \sum_{k \in I_j} \abs{\braket{\phi_k}{\psi_0}}^2$ for all $j = 1, \dots, M$. Then we have $\sum_{j = 1}^M \beta_j = 1$. Also, without loss of generality, assume that $\lambda_1 \ge \lambda_2 \ge \cdots \ge \lambda_N$. Then for $h \ne j$ and any $k \in I_h$ and $\ell \in I_j$ we have $\abs{\lambda_k - \lambda_\ell} \ge \abs{h - j} \Delta$ where $\Delta$ is defined in \eqref{equ:eigen-dist}. Putting the above together we have
\begin{align*}
    \opnorm{P_T(\cdot \vert \psi_0) - P_\infty(\cdot \vert \psi_0)}_1
    & = \sum_{v \in V} \abs{P_T(v \vert \psi_0) - P_\infty(v \vert \psi_0)} \\
    & \le \sum_{h, j = 1 : h \ne j}^M \frac{\beta_h + \beta_j}{\abs{h - j} T \Delta} \\
    & = \sum_{r = 1}^{M - 1} \frac{1}{rT\Delta} \sum_{h, j : \abs{h - j} = r} \beta_h + \beta_j \\
    & \le \frac{2\ln M + 2}{T\Delta},
\end{align*}
where in the last inequality we have used the bound $\sum_{i = 1}^n 1 / n \le \ln n + 1$ for harmonic sums, and the fact that the inner sum is always $\le 2$.

\end{document}